\documentclass[a4paper,11pt]{amsart}
\usepackage{amssymb}
\DeclareMathSymbol{\subsetneqq}{\mathbin}{AMSb}{36}
\newcommand{\R}{\mathbb{R}}

\newcommand{\C}{\mathbb{C}}

\newcommand{\beq}{\begin{eqnarray}}
\newcommand{\eeq}{\end{eqnarray}}
\newcommand{\bq}{\begin{equation}}
\newcommand{\eq}{\end{equation}}
\newcommand{\beqn}{\begin{eqnarray*}}
\newcommand{\eeqn}{\end{eqnarray*}}
\newcommand{\bex}{\begin{exo}}
\newcommand{\eex}{\end{exo}}
\newcommand{\ben}{\begin{enumerate}}
\newcommand{\een}{\end{enumerate}}
\newtheorem{th1}{{\bf Theorem}}[section]
\newtheorem{thm}[th1]{{\bf Theorem}}
\newtheorem{lem}[th1]{{\bf Lemma}}
\newtheorem{prop}[th1]{{\bf Proposition}}

\newtheorem{rem}[th1]{\bf Remark}

\newtheorem{defi}[th1]{\bf Definition}

\author[T. Saanouni]{Tarek Saanouni}
\address{University of Tunis El Manar, Faculty of Science of Tunis, LR03ES04 partial differential Equations and applications, 2092 Tunis, Tunisia.}
\email{\sl Tarek.saanouni@ipeiem.rnu.tn}
\thanks{{\sf  T. Saanouni} is grateful to the Laboratory of
PDE and Applications at the Faculty of Sciences of Tunis.}


\title[Remarks on the inhomogeneous fractional NLS]{Remarks on the inhomogeneous fractional nonlinear Schr\"odinger equation}

\date{\today}
\begin{document}
\begin{abstract}
Using a sharp Gagliardo-Nirenberg type inequality, well-posedness issues of the initial value problem for a fractional inhomogeneous Schr\"odinger equation are investigated.
\end{abstract}


\maketitle

\tableofcontents
\vspace{ 1\baselineskip}
\renewcommand{\theequation}{\thesection.\arabic{equation}}
\section{Introduction}
Consider the initial value problem for an inhomogeneous nonlinear Schr\"odinger equation
\begin{equation}\label{eq1}
\left\{
\begin{matrix}
i\dot u-(-\Delta)^\alpha u+\epsilon|x|^\gamma u|u|^{p-1}=0;\\
u_{|t=0}= u_0,
\end{matrix}
\right.
\end{equation}
which models various physical contexts in the description of nonlinear
waves such as propagation of a laser beam and plasma waves. For example, when
$\gamma = 0$, it arises in nonlinear optics, plasma physics and fluid mechanics \cite{adkm,baz}. When $\gamma> 0$, it can be thought
of as modeling inhomogeneities in the medium. The nonlinearity enters due to the
effect of changes in the field intensity on the wave propagation characteristics of
the medium and the nonlinear weight can be looked as the proportional to the
electron density \cite{g,lt,tw}.\\

Here and hereafter $\epsilon\in\{\pm1\}$, $N\geq2$, $\alpha\in(0,1)$, the inhomogeneous exponent $\gamma\in\R$ and $u(t,x) : \R_+\times\R^N\to\C$. 
The following quantities are called respectively mass and energy are conserved under the flow of the problem \eqref{eq1}.
\begin{gather*}
M(t)=M(u(t)):=\int_{\R^N}|u(t)|^2\,dx;\\
E(t)=E(u(t)):=\int_{\R^N}\Big(\frac12|(-\Delta)^\frac\alpha2u(t)|^2- \frac\epsilon{1+p}|x|^\gamma|u(t)|^{1+p}\Big)dx.
\end{gather*}
If $\epsilon=-1$, the energy is positive and \eqref{eq1} is said to be defocusing, otherwise it is focusing and a control of the $\|\,.\,\|_{H^\alpha}$ norm of a local solution is no longer possible with the conserved laws.\\

\noindent In the classical case $\alpha=1$ and $\gamma=0$, for $1<p\leq\frac{N+2}{N-2}$ if $N\geq 3$ and $1<p<\infty$ if $N\in\{1,2\}$, local well-posedness in the energy space holds \cite{J.B1,J.B2}. In the defocusing case, the solution to the Cauchy problem \eqref{eq1} is global and scatters if $p> 1+\frac4N$. In the focusing sign, the solution is global if $p<1+\frac4N$ or $p=1+\frac4N$ with small data \cite{sl}. Moreover, when $1<p<\frac{N+2}{N-2}$ and $\varepsilon=1$, there exists ground state which is stable if $p<1+\frac4N$ and unstable if $p\geq1+\frac4N$ \cite{sl}.\\

When $\gamma=0$ and $\alpha\in(0,1)$, the problem \eqref{eq1} is a nonlocal model known as nonlinear fractional Schr\"odinger equation which has also attracted much attentions recently \cite{chho,chho2,chkl,chkl2,ghx,gh,gh2,gh3,gw,gs,fl,fls}. The fractional Schr\"odinger equation is a fundamental equation of fractional quantum mechanics, which was derived by Laskin \cite{l,l2} as a result of extending the Feynman path integral, from the Brownian-like to Levy-like quantum mechanical paths. It is proved that the Cauchy problem is well-posed and scatters in the radial energy space \cite{gh,gh2}.\\

In the classical Laplacian case $\alpha=1$, if $\gamma>0$ the problem \eqref{eq1} was treated in \cite{jc}, where well-posedness was discussed via potential-well method, a recent work \cite{hms} considers a similar problem with mixed power nonlinearity. If $\gamma<0$, the stability issues for ground states were investigated by many authors \cite{m,fw,lww,fo}. \\

This note seems be the first one dealing with well-posedness issues for the inhomogeneous fractional Schr\"odinger equation.\\

The purpose of this manuscript is twice. First, some classical Gagliardo-Nirenberg type inequality is generalized to the inhomogeneous fractional case. Second, some analogue results to the classical Schr\"odinger case $\alpha=1$ about global well-posedness of \eqref{eq1} are obtained in the radial case.\\

\noindent The rest of the paper is organized as follows. The second section is devoted to give the main results and some tools needed in the sequel. Sections three, four and five deal with some Gagliardo-Nirenberg inequality and it's best constant. The sixth section contains a proof of global well-posedness for the Schr\"odinger problem \eqref{eq1}. Section seven deals with existence and stability of ground states. The last section is about global well-posedness of the Schr\"odinger problem \eqref{eq1} in the focusing case using potential well method \cite{ps}. In the appendix, some compact Sobolev injection is proved.\\

Here and hereafter, $C$ denotes a constant which may vary from line to line, $A\lesssim B$ denotes an estimate of the form $A\leq C B$
for some absolute constant $C$, $\int .\,dx:=\int_{\R^N}.\,dx$, $L^p:=L^p({\R^N})$ is the Lebesgue space endowed with the norm $\|\,.\,\|_p:=\|\,.\,\|_{L^p}$ and $\|\,.\,\|:=\|\,.\,\|_2$. The classical fractional Sobolev space is $H^{\alpha,p}:=(I-\Delta)^{-\frac\alpha2}L^p$ and $H^\alpha:=H^{\alpha,2}$ is the energy space endowed with the complete norm 
$$\|u\|_{H^\alpha}:=(\|u\|^2+\|(-\Delta)^\frac\alpha2 u\|^2)^\frac12.$$
 If $T>0$ and $X$ is an abstract functional space, $C_T(X):=C([0,T],X)$, $L_T^p(X):=L^p([0,T],X)$ and $X_{rd}$ is the set of radial elements in $X$. Finally, for an eventual solution of \eqref{eq1}, $T^*>0$ denotes it's lifespan. 
 \section{Main results and background}
In this section we give the main results and some technical tools needed in the sequel. Here and hereafter, we denote, for $\phi\in H^\alpha_{rd},$ the so-called action 
$$S(\phi):=E(\phi)+\frac12 M(\phi)= \frac{1}{2} \| \phi \| ^{2}_{H^{\alpha}} -\frac{1}{p+1}\int |x|^\gamma|\phi|^{p+1}\,dx.$$
If $\lambda\neq0$ and $(a,b)\in\R^2$, we define the scaling $\phi_{a,b}^\lambda:=\lambda^a\phi(\frac.{\lambda^b})$ and the operators
\begin{gather*}
 \mathcal{L}_{a,b}S(\phi):=\partial_{\lambda}(S(\phi^{\lambda}_{a,b}))_{|
\lambda=1};\\
 K_{a,b}:=\mathcal{L}_{a,b} S\quad\mbox{and}\quad H_{a,b}:=S-\frac1{2a+Nb}K_{a,b}.
\end{gather*}
We denote the functional defined on $H^\alpha_{rd}$ by
$$J(u):=\frac{\|(-\Delta)^\frac\alpha2u\|^B\|u\|^A}{\int|x|^\gamma|u|^{1+p}\,dx},$$
where
\begin{gather*}
A:=(1-\mu)(1+p-\frac{2\gamma}{N-2\alpha})=1+p-\frac1{2\alpha}(N(p-1)-2\gamma);\\ B:=\mu(1+p-\frac{2\gamma}{N-2\alpha})+\frac{2\gamma}{N-2\alpha}=\frac1{2\alpha}(N(p-1)-2\gamma);\\
\mu:=\frac N\alpha(\frac12-\frac1{p+1-\frac{2\gamma}{N-2\alpha}}).
\end{gather*}
\subsection{Main results}
Results proved in this paper are listed in what follows. First, we derive an inhomogeneous Gagliardo-Nirenberg type inequality.
\begin{thm}\label{t0}
Let $\gamma\in\R$, $\alpha\in(0,1)$ and $2\leq p+1-\frac{2\gamma}{N-2\alpha}\leq\frac{2N}{N-2\alpha}$. Then, there exists a positive constant $C(N,p,\gamma,\alpha)$, such that for any $u\in H^\alpha_{rd}$,
\begin{equation}\label{ineq}
\int|x|^\gamma|u|^{1+p}\,dx\leq C(N,p,\gamma,\alpha)\|u\|^A\|u\|_{\dot H^\alpha}^B.\end{equation}
\end{thm}
Next, we are concerned with the best constant $C(N,p,\gamma,\alpha)$ in the previous inequality. 
\begin{thm}\label{t0'}
Let $\gamma\in\R$, $\alpha\in(0,1)$ and $2<p-\frac{2\gamma}{N-2\alpha}<\frac{2N}{N-2\alpha}$. Then, 
$$\beta:=\inf\Big\{J(u),\quad u\in H^\alpha_{rd}\Big\}$$
is attained in some $\psi\in H^\alpha_{rd}$ satisfying
\begin{equation}\label{euler}
B(-\Delta)^\alpha\psi+A\psi-\beta(p+1)|x|^\gamma\psi|\psi|^{p-1}=0
\end{equation}
and $\beta=(\int|x|^\gamma|\psi|^{p+1}\,dx)^{-1}$.
\end{thm}
\begin{thm}\label{t0''}
Let $\gamma\in\R$, $\alpha\in(0,1)$ and $2<p-\frac{2\gamma}{N-2\alpha}<\frac{2N}{N-2\alpha}$. Then, 
$$C(N,p,\gamma,\alpha)=\frac{1+p}{{A}}(\frac AB)^{\frac{B}2}\|\phi\|^{-(p-1)},$$
where $\phi$ is a ground state solution to 
$$(-\Delta)^\alpha \phi +\phi -|x|^\gamma\phi|\phi|^{p-1}=0 ,\quad 0\neq \phi \in H^\alpha_{rd}.$$
\end{thm}
Using Theorem \ref{t0}, well-posedness of the Schr\"odinger problem \eqref{eq1} in the energy space holds.
\begin{prop}\label{t1}
Let $\gamma\in\R$, $\alpha\in(\frac N{2N-1},1)$, $2\leq p+1-\frac{2\gamma}{N-2\alpha}\leq\frac{2N}{N-2\alpha}$ and $u_0\in H^\alpha_{rd}$. Then,
there exists $T^*>0$ and a unique maximal solution to \eqref{eq1},
$$u\in C([0,T^*),H^\alpha_{rd}).$$
Moreover,
\begin{enumerate}
\item
$T^*=\infty$ in the defocusing case;
\item
the mass and the energy are conserved.
\end{enumerate} 
\end{prop}
In the mass subcritical case, the local solution given by the previous result is global.
\begin{prop}\label{sscrtq}
Let $\gamma\in\R$, $\alpha\in(\frac N{2N-1},1)$, $2\leq p+1-\frac{2\gamma}{N-2\alpha}\leq\frac{2N}{N-2\alpha}$ and $u_0\in H^\alpha_{rd}$. Then, the solution given by the previous result is global if one of the next conditions holds.
\begin{enumerate}
\item
$B<2$;
\item
$B=2$ and $M(0)<\Big(\frac{p+1}{2C(N,p,\gamma,\alpha)}\Big)^\frac2 A$.
\end{enumerate}
\end{prop}
Now, we are concerned with the focusing case $\epsilon=1$. Indeed, using the potential well method due to Payne-Sattinger \cite{ps}, we extend the local solution with data in some stable sets to a global one. Here, we are reduced to prove that the fractional elliptic problem 
$$(-\Delta)^\alpha\phi+\phi-|x|^\gamma\phi|\phi|^{p-1}=0,\quad 0\neq \phi \in H^{\alpha}_{rd}$$
has a ground state in the meaning that it has a nontrivial positive radial solution which minimizes the problem
\begin{equation}\label{m1}
 m_{a,b}:=\inf_{0\neq \phi \in H^{\alpha}_{rd}}\Big\{ S(\phi),   \;s.\,\, t. \quad  K_{a,b}(\phi) =0\Big\}.
\end{equation}
Let us recall the stability of ground state.
\begin{defi}
In the focusing case $(\epsilon=1)$, if $\phi$ is a ground state of \eqref{eq1}, then $e^{it}\phi\in C(\mathbb R,H^\alpha)$ is a global solution to \eqref{eq1} called standing wave. This standing wave is said to be stable if
$$\forall\varepsilon>0,\,\,\exists\delta>0\quad\mbox{ s. t. }\quad\Big(u_0\in U_\delta(\phi)\Rightarrow u(t)\in U_\varepsilon(\phi),\,\,\forall t\in(0,T^*)\Big),$$
where $u\in C_{T^*}(H^\alpha)$ is the solution of \eqref{eq1} and
$$U_\delta(\phi):=\Big\{v\in H^\alpha_{rd}\quad\mbox{ s. t. }\quad \inf_{\theta\in\mathbb R}\|v-e^{i\theta}\phi\|_{H^\alpha}<\delta\Big\},$$
otherwise, it is unstable.
\end{defi}
The next result guarantees the existence of ground state.
\begin{thm}\label{t3}
Take $\epsilon=1$, $\gamma\geq0$, $\alpha\in(0,1)$, a couple of real numbers $(a,b)\in \mathbb{R}_{+}^{*}\times
\mathbb{R}_{+}$ and $2 <p+1-\frac{2\gamma}{N-2\alpha}<\frac{2N}{N-2\alpha}.$ Then,
\begin{enumerate}
\item[(1)] $ m := m_{a,b} $ is nonzero and independent of $(a,b)$;
\item[(2)] there is a ground state solution to \eqref{eq1} in the following meaning
\begin{equation}\label{gst}
(-\Delta)^\alpha \phi +\phi -|x|^\gamma\phi|\phi|^{p-1}=0 ,\quad 0\neq \phi \in
 H_{rd}^{\alpha},\quad m=S(\phi);
\end{equation}
\item[(3)]this ground state is orbitally stable if $B<2$.
\end{enumerate}
\end{thm}
Now, we give some invariant sets under the flow of \eqref{eq1}, which yield to a global solution. We denote, for $(a,b)\in\R_+^*\times\R_+$, the set
$$A_{a,b}:=\{\phi\in H^\alpha_{rd},\quad\mbox{s. t}\quad S(\phi)<m_{a,b}\quad\mbox{and}\quad K_{a,b}(\phi)>0\}.$$
\begin{thm}\label{t4}
Take $\epsilon=1$. Let $\gamma\geq0$, $\alpha\in(\frac N{2N-1},1)$, $2<p+1-\frac{2\gamma}{N-2\alpha}<\frac{2N}{N-2\alpha},$ $a>0$, $b\geq0$, $u_0\in A_{a,b}$ and $u\in C([0,T^*),H^\alpha_{rd})$ the maximal solution to \eqref{eq1}. Then $u$ is global and $u(t)\in A_{a,b}$ for any time $t\geq0$.
\end{thm}
\subsection{Tools}
Let us collect some classical results needed along this manuscript. We start with some properties of the free Schr\"odinger kernel.
\begin{prop}\label{fre}
Denoting the free operator associated to the fractional Schr\"odinger equation
$$T_{\alpha}(t)\phi:=e^{it(-\Delta)^\frac\alpha2}\phi:=\mathcal F^{-1}(e^{-it|y|^{2\alpha}})*\phi,$$
yields
\begin{enumerate}
\item
$T_{\alpha}(t)u_0-i\epsilon\int_0^tT_{\alpha}(t-s)[|x|^\gamma u|u|^{p-1}]\,ds$ is the solution to the problem \eqref{eq1};
\item
$(T_{\alpha}(t))^*=T_{\alpha}(-t)$;
\item
$T_{\alpha}T_{\beta}=T_{\alpha+\beta}$;
\item
$T_{\alpha}(t)$ is an isometry of $L^2$.
\end{enumerate}
\end{prop}
\begin{defi}
A couple of real numbers $(q,r)$ such that $q,r\geq2$ is said to be admissible if 
$$\frac{4N+2}{2N-1}\leq q\leq\infty,\quad \frac2q+\frac{2N-1}r\leq N-\frac12,$$
or 
$$2\leq q\leq\frac{4N+2}{2N-1},\quad \frac2q+\frac{2N-1}r< N-\frac12.$$
\end{defi}
Recall the so-called Strichartz estimate \cite{gw}.
\begin{prop}\label{prop2}
Let $N \geq 2$, $\mu\in\R$, $\frac{N}{2N-1}<\alpha<1$ and $u_0\in H^\mu_{rd}$. Then
$$\|u\|_{L^q_t(L^r)\cap L^\infty_t(\dot H^\mu)}\lesssim\|u_0\|_{\dot H^\mu}+\|i\dot u-(-\Delta)^\alpha u\|_{L^{\tilde q'}_t(L^{\tilde r'})},$$
if $(q, r)$ and $(\tilde q,\tilde r)$ are admissible pairs such that $(\tilde q,\tilde r, N)\neq (2,\infty, 2)$ or $(q, r, N)\neq (2,\infty, 2)$ and satisfy the condition
$$\frac{2\alpha}q+\frac Nr=\frac N2-\mu,\quad \frac{2\alpha}{\tilde q}+\frac N{\tilde r}=\frac N2+\mu.$$
\end{prop}
\begin{rem}
If we take $\mu=0$ in the previous inequality, we obtain the classical Strichartz estimate.
\end{rem}
The  following fractional Gagliardo-Nirenberg inequality \cite{hyz,hmon} will be useful.
\begin{lem}\label{gn}
Let $\alpha\in(0,1)$, $2\leq p\leq\frac{2N}{N-2\alpha}$ and $\theta:=\frac N\alpha(\frac12-\frac1p).$ Then
$$\|u\|_{{p}}\lesssim\|u\|^{1-\theta}\|u\|_{\dot H^{\alpha}}^\theta,$$
for any $u\in H^\alpha(\R^N)$.
\end{lem}
We give also an estimate similar to Strauss \cite{st} inequality in the fractional case \cite{co}.
\begin{lem}\label{strs}
Let $N\geq2$ and $\frac12<\alpha<\frac N2$. Then
\begin{equation}\label{straus}
\sup_{x\neq0}|x|^{\frac N2-\alpha}|u(x)|\leq C(N,\alpha)\|(-\Delta)^\frac\alpha2u\|,
\end{equation}
for any $u\in \dot H^\alpha(\R^N)$, where
$$C(N,\alpha)=\Big(\frac{\Gamma(2\alpha-1)\Gamma(\frac N2-\alpha)\Gamma(\frac N2)}{2^{2\alpha}\pi^{\frac N2}\Gamma^2(\alpha)\Gamma(\frac N2-1+\alpha)}\Big)^\frac12$$
and $\Gamma$ is the Gamma function.
\end{lem}
Sobolev injection \cite{AC1,pl} gives a meaning to the energy and several computations done in this note.
\begin{lem}\label{sblv}
Let $N\geq2$, $p\in(1,\infty)$ and $\alpha\in(0,1)$, then 
\begin{enumerate}
\item
$H^\alpha \hookrightarrow L^q$ for any $q\in[2,\frac{2N}{N-2\alpha}]$;
\item
the following injection $H^\alpha_{rd} \hookrightarrow\hookrightarrow L^q$ is compact for any $q\in(2,\frac{2N}{N-2\alpha})$.
\end{enumerate}
\end{lem}
\begin{defi}
We define the space
$$\Sigma:=\{u\quad\mbox{radial and measurable in}\quad \R^N,\quad\mbox{s. t.}\quad u\in L^{1+p}(|x|^\gamma\,dx)\}$$
endowed with the norm 
$$\|u\|_\Sigma:=\Big(\int|x|^\gamma|u(x)|^{1+p}\,dx\Big)^\frac1{1+p}.$$
\end{defi}
Finally, we give a compact Sobolev injection known in the classical case \cite{JB,kw}.
\begin{lem}\label{cmpct}
Let $\gamma\in\R$, $\alpha\in(0,1)$ and $2 <p+1-\frac{2\gamma}{N-2\alpha}<\frac{2N}{N-2\alpha}.$ Then, the following injection is compact
\begin{equation}\label{cmpctt}
H^\alpha_{rd}(\R^N)\hookrightarrow\hookrightarrow\Sigma.\end{equation}
\end{lem}
For the reader convenience, we give a proof in the appendix.
\section{Proof of Theorem \ref{t0}}
In this section, we prove the interpolation inequality \eqref{ineq}. First, using Lemma \ref{strs}, we get
\begin{eqnarray*}
\int|x|^\gamma|u(x)|^{1+p}\,dx
&=&\int(|x|^{\frac N2-\alpha}|u(x)|)^{\frac{2\gamma}{N-2\alpha}}|u(x)|^{1+p-\frac{2\gamma}{N-2\alpha}}\\
&\lesssim&\|u\|_{\dot H^\alpha}^{\frac{2\gamma}{N-2\alpha}}\int|u(x)|^{1+p-\frac{2\gamma}{N-2\alpha}}.
\end{eqnarray*}
Now, thanks to Lemma \ref{gn}, yields
\begin{eqnarray*}
\int|x|^\gamma|u(x)|^{1+p}\,dx
&\lesssim&\|u\|_{\dot H^\alpha}^\frac{2\gamma}{N-2\alpha}\|u\|^{1+p-\frac{2\gamma}{N-2\alpha}}_{1+p-\frac{2\gamma}{N-2\alpha}}\\
&\lesssim&\|u\|_{\dot H^\alpha}^\frac{2\gamma}{N-2\alpha}(\|u\|^{1-\mu}\|u\|_{\dot H^\alpha}^\mu)^{1+p-\frac{2\gamma}{N-2\alpha}}\\
&\lesssim&\|u\|^{(1-\mu)(1+p-\frac{2\gamma}{N-2\alpha})}\|u\|_{\dot H^\alpha}^{\mu(1+p-\frac{2\gamma}{N-2\alpha})+\frac{2\gamma}{N-2\alpha}}.
\end{eqnarray*}
The proof is ended.
\section{Proof of Theorem \ref{t0'}}
Using Theorem \ref{t0}, there exists a sequence $(v_n)$ in $H^{\alpha}_{rd}$ such that 
$$\beta=\lim_nJ(v_n).$$
Denoting for $a,b\in\R$, the scaling $u^{a,b}:=a u(b .)$, we compute
\begin{gather*}
\|(-\Delta)^\frac\alpha2u^{a,b}\|^2=a^2b^{2\alpha-N}\|(-\Delta)^\frac\alpha2u\|^2;\\
\|u^{a,b}\|^2=a^2b^{-N}\|u\|^2;\\
\int|x|^\gamma|u^{a,b}(x)|^{1+p}\,dx=a^{1+p}b^{-N-\gamma}\int|x|^\gamma|u|^{1+p}(x)\,dx.
\end{gather*}
It follows that 
$$J(u^{a,b})=J(u).$$
Now, we choose 
$$\mu_n:=\Big(\frac{\|v_n\|}{\|(-\Delta)^\frac\alpha2 v_n\|}\Big)^\frac1\alpha\quad\mbox{and}\quad \lambda_n:=\frac{\|v_n\|^{\frac N{2\alpha}-1}}{{\|(-\Delta)^\frac\alpha2 v_n\|}^\frac N{2\alpha}}.$$
Thus, $\psi_n:=v_n^{\lambda_n,\mu_n}$ satisfies
$$\|\psi_n\|=\|(-\Delta)^\frac\alpha2\psi_n\|=1\quad\mbox{and}\quad \beta=\lim_nJ(\psi_n).$$
Then, $\psi_n\rightharpoonup\psi$ in $H^\alpha_{rd}$ and using Sobolev injection \eqref{cmpctt}, we get for a subsequence denoted also $(\psi_n)$, $$\int|x|^\gamma|\psi_n|^{1+p}\,dx\rightarrow\int|x|^\gamma|\psi|^{1+p}\,dx.$$
This implies that, when $n$ goes to infinity
$$J(\psi_n)=\frac1{\int|x|^\gamma|\psi_n|^{1+p}\,dx}\rightarrow\frac1{\int|x|^\gamma|\psi|^{1+p}\,dx}.$$
Using lower semi continuity of the $H^\alpha$ norm, we get 
$$\|\psi\|\leq1\quad\mbox{and}\quad \|(-\Delta)^\frac\alpha2\psi\|\leq1.$$
Then, $J(\psi)< \beta$ if $\|\psi\|\|(-\Delta)^\frac\alpha2\psi\|<1$, which implies that
$$\|\psi\|=1\quad\mbox{and}\quad \|(-\Delta)^\frac\alpha2\psi\|=1.$$
It follows that $$\psi_n\rightarrow\psi\quad\mbox{in}\quad H^\alpha$$
and
$$\beta=J(\psi)=\frac1{\int|x|^\gamma|\psi|^{1+p}\,dx}.$$
The minimizer satisfies the Euler equation
$$\partial_\varepsilon J(\psi+\varepsilon\eta)_{|\varepsilon=0}=0,\quad\forall \eta\in C_0^\infty\cap H^\alpha_{rd}.$$
Hence $\psi$ satisfies \eqref{euler}. This completes the proof.
\section{Proof of Theorem \ref{t0''}}
Thanks to Theorem \ref{t0'}, we know that
$C(N,p,\gamma,\alpha)=\frac1\beta=\int|x|^\gamma|\psi(x)|^{1+p}\,dx$, where $\psi$ is given in Theorem \ref{t0'}.
Take, for $a,b\in\R$, the scaling $\psi=\phi^{a,b}:=a\phi(b.)$. Then, the fact that
$$B(-\Delta)^\alpha\psi+A\psi-\beta(p+1)|x|^\gamma\psi|\psi|^{p-1}=0,$$
implies that
$$Aa\Big(-\frac BAb^{2\alpha}(-\Delta)^\alpha\phi+\phi-\frac\beta A(p+1)a^{p-1}b^{-\gamma}|x|^\gamma\phi|\phi|^{p-1}\Big)=0.$$
Choosing
$$b=\Big(\frac AB\Big)^\frac1{2\alpha}\quad\mbox{and}\quad a=\Big((\frac{A}B)^\frac\gamma{2\alpha}\frac A{\beta(1+p)}\Big)^\frac1{p-1},$$
it follows that 
$$-(-\Delta)^\alpha\phi+\phi-|x|^\gamma\phi|\phi|^{p-1}=0.$$
Now, since 
$$\|\psi\|=1=ab^{-\frac N2}\|\phi\|,$$
we get
$$\beta=\frac{A}{1+p}(\frac AB)^{-\frac B2}\|\phi\|^{p-1}.$$
The proof is closed.
\section{Well-posedness}
\subsection{Proof of Proposition \ref{t1}}
The proof follows using Theorem \ref{t0} via Strichartz estimate in Proposition \ref{prop2}, the integral formula in Proposition \ref{fre} and a classical Picard fixed point method,  arguing like \cite{Cas1,gw}.
\subsection{Proof of Proposition \ref{sscrtq}}
Assume that $B<2$ or $B=2$ and $M(0)<\Big(\frac{p+1}{2C(N,p,\gamma,\alpha)}\Big)^\frac2A$. With contradiction, assume that $T^*=\infty$. Then $$\limsup_{T^*}\|u(t)\|_{\dot H^\alpha}=\infty.$$
Write, using previous notations and Theorem \ref{t0},
\begin{eqnarray*}
2E(t)
&=&\|u\|_{\dot H^\alpha}^2-\frac2{p+1}\int|x|^\gamma|u|^{1+p}\,dx\\
&\geq&\|u\|_{\dot H^\alpha}^2-\frac{2C(N,p,\gamma,\alpha)}{p+1}\|u\|^{A}\|u\|_{\dot H^\alpha}^{B}\\
&\geq&\|u\|_{\dot H^\alpha}^2\Big(1-\frac{2C(N,p,\gamma,\alpha)}{p+1}M(0)^{\frac A2}\|u\|_{\dot H^\alpha}^{B-2}\Big).
\end{eqnarray*}
This contradicts the infinite limit above.
\section{Proof of Theorem \ref{t3}}
In this section we prove the existence of a ground state solution to \eqref{gst} which is stable for some low range of the nonlinearity exponent. With a direct computation, yields, for $a,b\in\R_+$ and $\phi\in H^\alpha$,
\begin{gather*}
K_{a,b}(\phi)=\frac12(2a+Nb)\|\phi\|^2+\frac12(2a+(N-\alpha)b)\|(-\Delta)^\frac\alpha2\phi\|^2-(a+\frac{b(\gamma+N)}{1+p})\int|x|^\gamma|\phi|^{1+p}dx;\\
H_{a,b}(\phi)=\frac{\alpha b}{2(2a+Nb)}\|(-\Delta)^\frac\alpha2\phi\|^2+\frac{b\gamma+a(p-1)}{(1+p)(2a+Nb)}\int|x|^\gamma|\phi|^{1+p}dx.
\end{gather*}
Define, the quadratic and nonlinear parts of $K_{a,b}$ as follows
$$K_{a,b}^Q(\phi)=\frac12(2a+Nb)\|\phi\|^2+\frac12(2a+(N-\alpha)b)\|(-\Delta)^\frac\alpha2\phi\|^2,\quad K^N_{a,b}:=K_{a,b}-K^Q_{a,b}.$$
\begin{rem}{Note that, in this section}
\begin{enumerate}
\item[1)] $(a,b) \in \mathbb{R}^*_{+}\times\R_+$;
\item[2)] the proof of the Theorem \ref{t3} is based on several lemmas;
\item[3)] we write, for easy notation, $\phi^\lambda:=\phi^\lambda_{a,b}$, $K=K_{a,b}$, $K^{Q}=K_{a,b}^{Q}$, $K^{N}=K_{a,b}^{N}$, $\mathcal{L}=\mathcal{L}_{a,b}$ and $H=H_{a,b}.$
\end{enumerate}
\end{rem}
\subsection{Existence of ground state}
\begin{lem}  Let $0\neq \phi \in H^{\alpha}_{rd},$ then 
\begin{enumerate}
 \item[1)] $min(\mathcal{L} H(\phi),H(\phi)) > 0$;
 \item[2)] $\lambda \mapsto H(\phi^{\lambda})$ is increasing.
\end{enumerate}
\end{lem}
\begin{proof}
Denoting $\underline{\mu}:=2a +Nb$ and $\bar\mu:=\mu-\alpha b$, we compute
\begin{eqnarray*}
\mathcal{L}(H(\phi))
&=&\mathcal L(1-\frac{\mathcal L}{\underline\mu})S(\phi)\\
&=&{\bar\mu} H(\phi) + \frac{1}{\underline{\mu}} (\mathcal{L}- \bar\mu)(\underline{\mu}-\mathcal{L})S(\phi)\\
&\geq& \frac{1}{\underline{\mu}} (\mathcal{L}- \bar\mu)(\underline{\mu}-\mathcal{L})S(\phi).
\end{eqnarray*}
 Since $(\mathcal{L}-\bar\mu) \|(-\Delta)^{\frac\alpha2}\phi\|^{2}
=0=(\mathcal{L}-\underline{\mu}) \| \phi\|^{2},$ we have
$(\mathcal{L}-\bar\mu)(\mathcal{L}-\underline{\mu}) \| \phi
\|_{H^{\alpha}}^{2}=0.$ Then, with a direct computation
\begin{eqnarray*}
\mathcal{L}(H(\phi))
&\geq& \frac{1}{\underline{\mu}} (\mathcal{L}- \bar\mu)(\underline{\mu}-\mathcal{L})S(\phi)\\
&=&\frac{(a(p-1)+b\gamma)(a(p-1)+b\gamma+\alpha b)}{\underline{\mu}(p+1)} \int|x|^\gamma|u|^{p+1}\,dx\\
&>&0.
\end{eqnarray*}
The second point is a consequence of the equality $\partial_{\lambda}
H(\phi^{\lambda})=\mathcal{L}H(\phi^{\lambda}).$
\end{proof}
The next auxiliary result reads.
\begin{lem}\label{case}
Assume that $\gamma\geq0$, $ 2a+(N-\alpha)b \neq 0$ and take $(\phi_{n})$ a
bounded sequence of $H^{\alpha}_{rd}-\{ 0\}$ satisfying $\displaystyle\lim_{n\rightarrow+\infty} K^{Q}(\phi_{n})=0.$ Then, there exists
$n_{0} \in \mathbb{N}$ such that $K(\phi_{n}) >0 $ for all $n\geq n_{0}.$
\end{lem}
\begin{proof}
Since $a,b \geq 0 $, $2a +(N-\alpha)b \neq0$ and
$$K^{Q}(\phi_{n})=\frac{2a +(N-\alpha)b}{2} \|
 (-\Delta)^{\frac\alpha2} \phi_{n}\|^{2} + \frac{ (2a +Nb)}{2} \|
 \phi_{n}  \|^{2}\geq C  \|\phi_{n}\|^{2}_{H^{\alpha}},$$
we get
$$\lim_{n\rightarrow\infty}\|\phi_{n}\|_{H^\alpha}=0.$$
Taking into account of Theorem \ref{t0} and the fact that $2<1+p-\frac{2\gamma}{N-2\alpha}<\frac{N+2\alpha}{N-2\alpha}$, we get 
$$K^{N}(\phi_{n})\lesssim \|\phi_{n}\|^{1+p}_{H^{\alpha}}=o(K^Q(\phi_n)).$$ 
It follows that, when $n$ goes to infinity
$$K(\phi_{n})\simeq K^{Q}(\phi_{n})>0.$$
The proof is finished.
\end{proof}
The last intermediary result is the following.
\begin{lem}\label{lemma 3}
We have
$$m_{a,b}=\inf_{0\neq \phi \in H^{\alpha}_{rd}} \Big\{H(\phi),\; s.\,t.
\quad K(\phi) \leq 0\Big\}.$$
\end{lem}
\begin{proof}
It is sufficient to prove that $m_{a,b} \leq m_{1}$, where $m_{1}$ is the right hand side of the previous equality. Take $\phi \in H^{\alpha}_{rd}$ such that $K(\phi) < 0$. Because $2a +(N-\alpha)b\neq 0,$ by the previous lemma, the facts that
$\displaystyle\lim_{\lambda \rightarrow0} K^{Q}(\phi^{\lambda})=0$ and
$\lambda \mapsto H(\phi^{\lambda})$ is increasing, there exists
$\lambda\in(0,1)$ such that
\begin{equation}\label{Kphi}
K(\phi^{\lambda})=0 \quad\mbox{and}\quad H(\phi^{\lambda}) \leq H(\phi).
\end{equation}
Then, $m_{a,b}\leq H(\phi^{\lambda}) \leq H(\phi)$. This ends the proof.
\end{proof}
\begin{proof}[\bf Proof of Theorem \ref{t3}]
The proof contains four steps.\\ 
$\textbf{Step 1}.$ A minimizing sequence is bounded in $H^{\alpha}_{rd}.$ \\
Let $(\phi_{n})$ be a minimizing sequence of \eqref{m1}, namely
\begin{equation}\label{minim}
0\neq\phi_{n}\in H^{\alpha}_{rd}, \quad K(\phi_{n})=0 \quad and \quad \lim_{n}H(\phi_{n})=\lim_{n }S(\phi_{n})=m.\end{equation}
$\bullet$ First case $b\neq0.$ Since, when $n$ tends to infinity
$$\frac{\alpha b}{2(2a+Nb)}\|(-\Delta)^{\frac\alpha2}\phi_{n}\|^{2}\leq H(\phi_{n})\rightarrow m,$$
we get 
$$\sup_n\|\phi_{n}\|_{\dot H^\alpha}\lesssim 1.$$
Assume that $\displaystyle\lim_n\|\phi_{n}\|=\infty$. Using the equality $K(\phi_{n})=0$ and Theorem \ref{t0}, yields
\begin{eqnarray*}
\|\phi_{n}\|^{2}
&\lesssim&\frac{2a +(N-2)b}{2} \|(-\Delta)^\frac\alpha2 \phi_{n}
\|^{2} +\frac{2 a +(N-\alpha)b}{2}\|\phi_{n}\|^{2}\\
&=&\frac{a(p+1)+(N+\gamma)b}{p+1}\int|x|^\gamma|\phi_{n}|^{p+1}\,dx \\
&\lesssim&\|\phi_{n}\|^A\|\phi_{n}\|_{\dot H^\alpha}^B\\
&\lesssim&\|\phi_{n}\|^A.
\end{eqnarray*}
The fact that $A=\frac N\alpha+(1+p-\frac{2\gamma}{N-2\alpha})(1-\frac N{2\alpha})>2$ leads to a contradiction in the last inequality if letting ${n\mapsto +\infty }$. Then $(\phi_{n})$ is bounded in $H^{\alpha}_{rd}.$\\
$\bullet$ Second case $ b =0$.\\
In this case $(\phi_{n})$ is bounded in $H^{\alpha}_{rd}$ because
$$\|\phi_n\|_{H^\alpha}^2=\int|x|^\gamma|\phi_n|^{p+1}\,dx\lesssim H(\phi_n)\rightarrow m.$$
 $\textbf{Step 2}.$ The weak limit of $(\phi_n)$ is nonzero.\\
Using the first step, via compact Sobolev injection in Lemma \ref{sblv}, for a subsequence, still denoted by $(\phi_{n})$, we have
$$\phi_{n} \rightharpoonup \phi\quad\mbox{ weakly in }\quad H^{\alpha}_{rd}\quad\mbox{ and }\quad \phi_{n} \rightarrow \phi \quad\mbox{ in }\quad L^q,\quad\mbox{ for any }\quad 2<q<\frac{2N}{N-2\alpha}.$$
 We prove that $\phi\neq0$. Arguing by contradiction, assume that $\phi=0$. Since $2<p+1-\frac{2\gamma}{N-2\alpha}<\frac{2N}{N-2\alpha}$, using the estimate \eqref{straus}, we obtain
\begin{eqnarray*}
K^{Q}(\phi_{n})
&=&K^{N}(\phi_{n})\\
&\lesssim& \int(|x|^{\frac N2-\alpha}|\phi_n(x)|)^{\frac{2\gamma}{N-2\alpha}}|\phi_n|^{p+1-\frac{2\gamma}{N-2\alpha}}\,dx\\
&\lesssim& \int|\phi_n|^{p+1-\frac{2\gamma}{N-2\alpha}}\,dx\rightarrow 0.
\end{eqnarray*}
Since $2a+(N-\alpha) b >0$, thanks to lemma \ref{case}, there exists $n_{0}$ such that
$K(\phi_{n})>0,$ for all $n>n_{0}$, which  contradicts the fact that
$K(\phi_{n})=0.$  This implies that $\phi\neq 0.$\\
$\textbf{Step 3}.$ $\phi$ is a minimizer and $m>0.$\\
We have the convergence
\begin{eqnarray*}
\int|x|^\gamma|\phi_{n}-\phi|^{p+1}\,dx
&\lesssim& \int(|x|^{\frac N2-\alpha}|\phi_n-\phi|)^{\frac{2\gamma}{N-2\alpha}}|\phi_n-\phi|^{p+1-\frac{2\gamma}{N-2\alpha}}\,dx\\
&\lesssim& \int|\phi_n-\phi|^{p+1-\frac{2\gamma}{N-2\alpha}}\,dx\rightarrow 0.
\end{eqnarray*}
With the lower semi-continuity of $H^{\alpha}_{rd}$ norm, it follows that
\begin{eqnarray*}
0= \liminf_{n} K(\phi_{n})
&\geq&\frac{2a+(N-\alpha) b }{2} \liminf_n \|(-\Delta)^{\frac\alpha2}\phi_{n}
\|^{2} +\frac{2a+N b }{2} \liminf_n \|  \phi_{n}
\|^{2} \\
&-& \frac{a(p+1)+(N+\gamma)b }{p+1}\lim_n\int|x|^\gamma|\phi_{n}|^{p+1}\,dx \\
&\geq& \frac{2a+Nb }{2}\| \phi\|^{2}+ \frac{2a+(N-\alpha) b }{2} \|(-\Delta)^{\frac\alpha2}\phi
\|^{2} \\
&-&\frac{a(p+1)+(N+\gamma) b }{p+1}\int|x|^\gamma|\phi|^{p+1}\,dx\\
&=& K(\phi).
\end{eqnarray*}
Applying Fatou lemma, we obtain
\begin{eqnarray*}
m 
&\geq& \liminf_{n} H(\phi_{n})\\
&\geq&\frac{\alpha b }{2(2a +(N+\gamma) b) }\liminf_{n}\|(-\Delta)^{\frac\alpha2}\phi_{n}\|^{2}_{2}+ 
\frac{a(p-1) }{(p+1)(2a +N b )}\liminf_{n}\int|x|^\gamma|\phi_{n}|^{p+1}\,dx\\
&\geq& \frac{\alpha b }{2(2a +(N+\gamma) b) }\|(-\Delta)^{\frac\alpha2}\phi\|^{2}_{2}+\frac{a(p-1) }{(p+1)(2a +N b )}\int|x|^\gamma|\phi|^{p+1}\,dx\\
 &=&  H(\phi).
\end{eqnarray*}
Then $\phi$ satisfies
$$0\neq\phi\in H^{\alpha}_{rd},\quad K(\phi)\leq0\quad\mbox{ and }\quad H(\phi)\leq m.$$
By \eqref{Kphi}, we can assume that $\phi$ is a minimizer satisfying
$$0\neq\phi\in H^{\alpha}_{rd},\quad K(\phi)=0\quad\mbox{ and }\quad S(\phi)=H(\phi)=m.$$
Moreover
$$H(\phi)= \frac{\alpha b }{2(2a +N b) }\|(-\Delta)^{\frac\alpha2}\phi\|^{2}_{2}+\frac{a(p-1)+b\gamma }{(p+1)(2a +N b )}\int|x|^\gamma|\phi|^{p+1}\,dx> 0.$$
Thus
$$m>0.$$
$\textbf{Step 4}.$ $\phi$ is a ground state solution to \eqref{gst}.\\
Since $\phi$ satisfies \eqref{m1}, there is a Lagrange multiplier $\eta\in\mathbb{R}$ such that $S'(\phi)=\eta K'(\phi)$. Then
$$0=K(\phi)=\mathcal{L}S(\phi)=<S'(\phi),\mathcal{L}(\phi)>=\eta<K'(\phi),\mathcal{L}(\phi)>=\eta \mathcal{L}^{2}S(\phi).$$
Moreover, with previous computation, we have
\begin{eqnarray*}
-\mathcal{L}^{2} S(\phi) -\bar\mu\underline{\mu} S(\phi)
&=&-(\mathcal{L}-\bar\mu)(\mathcal{L}-\underline{\mu})S(\phi)\\
&=&\frac{(a(p-1)+b\gamma)(a(p-1)+\alpha b+b\gamma)}{p+1} \int|x|^\gamma|u|^{p+1}\,dx\\
&\geq& 0.
\end{eqnarray*}
Because $S(\phi)>0$, it follows that $\eta=0$ and $S'(\phi)=0$. Finally, $\phi$ is a ground state and $m$ is independent of $(a,b).$
\end{proof}
\subsection{Stability of ground state}
The proof proceeds by contradiction. Suppose that there exists a sequence $u_0^n\in H^\alpha$ such that, when $n$ goes to infinity
$$\|u^0_n-e^{it_n}\phi\|_{H^\alpha}\rightarrow0\quad\mbox{and}\quad\inf_{\theta\in\R}\|u_n(t_n)-e^{i\theta}\phi\|_{H^\alpha}>\varepsilon_0$$
for some sequence of positive real numbers $(t_n)$ and $\varepsilon_0>0$, where $u_n\in C([0,T^*),H^\alpha)$ is the solution to \eqref{eq1} with data $u_n^0$. Let us denote $\phi_n:=u_n(t_n)$. Because $\phi$ is a ground state to \eqref{eq1}, we have
$$S(\phi)=m\quad\mbox{and}\quad \|\phi\|:=q>0.$$
Thus
$$\|u_0^n\|\rightarrow q\quad\mbox{and}\quad S(u_0^n)\rightarrow m.$$
Indeed, by Theorem \ref{t0}, yields
\begin{eqnarray*}
\int|x|^\gamma|u_0^n-\phi|^{1+p}\,dx
&\lesssim&\|u_0^n-\phi\|_{H^\alpha}^{1+p}\rightarrow0.
\end{eqnarray*}
Using the conservation laws, it follows that
$$\|\phi_n\|\rightarrow q\quad\mbox{and}\quad S(\phi_n)\rightarrow m.$$
If $\phi_n$ has a subsequence converging to $\phi\in H^\alpha$, then 
$$\varepsilon_0<\inf_{\theta\in\R}\|\phi_n-e^{i\theta}\phi\|_{H^\alpha}\leq\|\phi_n-\phi\|_{H^\alpha}\rightarrow0.$$
This contradiction shows that it is sufficient to prove that any sequence $\phi_n\in H^\alpha$ satisfying 
$$\|\phi_n\|\rightarrow q\quad\mbox{and}\quad S(\phi_n)\rightarrow m$$
is relatively compact. We have
$$S(\phi_n)=\frac12\|\phi_n\|_{H^\alpha}^2-\frac1{1+p}\int|x|^\gamma|\phi_n|^{1+p}\,dx\rightarrow m,$$
so for some $\varepsilon>0$ and for large $n$,
\begin{eqnarray*}
m+\varepsilon
&\geq&\frac12\|\phi_n\|_{H^\alpha}^2-\frac1{1+p}\int|x|^\gamma|\phi_n|^{1+p}\,dx\\
&\geq&\frac12\|\phi_n\|_{H^\alpha}^2\Big(1-\frac{C(N,p,\gamma,\alpha)}{1+p}\|\phi_n\|^{A}\|\phi_n\|_{\dot H^\alpha}^{B-2}\Big).
\end{eqnarray*}
Since $B\leq2$, it follows that $\phi_n$ is bounded in $H^\alpha$. This finishes the proof.
\section{Proof of Theorem \ref{t4}}
The proof is based on two auxiliary results.
\begin{lem}\label{stb}
The set $A_{a,b}$ is invariant under the flow of \eqref{eq1}.
\end{lem}
\begin{proof}
Let $u_0\in A_{a,b}$ and $u\in C_{T^*}(H^\alpha)$ be the maximal solution to \eqref{eq1}. Assume that $u(t_0)\notin A_{a,b}^+$ for some time $t_0\in(0,T^*)$. Since the energy and the mass are conserved, we get $K_{a,b}(u(t_0))<0$. So, with a continuity argument, there exists a positive time $t_1\in(0,t_0)$ such that $K_{a,b}(u(t_1))=0.$  This contradicts the definition of $m$ and finishes the proof.
\end{proof}
\begin{lem}\label{stbb}
The set $A_{a,b}$ is independent of the couple $(a,b)$.
\end{lem}
\begin{proof}
Let $(a,b)$ and $(a',b')$ in ${\mathbb R}_+^*\times\R_+$. We denote, for $\delta\geq0$, the sets
\begin{gather*}
A_{a,b}^{+\delta}:=\{v\in H^\alpha\quad\mbox{s. t.}\quad S(v)<m-\delta\quad\mbox{and}\quad K_{a,b}(v)\geq0\};\\
A_{a,b}^{-\delta}:=\{v\in H^\alpha\quad\mbox{s. t.}\quad S(v)<m-\delta\quad\mbox{and}\quad K_{a,b}(v)<0\}.
\end{gather*}
By Theorem \ref{t3}, the reunion $A_{a,b}^{+\delta}\cup A_{a,b}^{-\delta}$ is independent of $(a,b)$. So, it is sufficient to prove that $A_{a,b}^{+\delta}$ is independent of $(a,b)$. If $S(v)<m$ and $K_{a,b}(v)=0$, then $v=0$. So, $A_{a,b}^{+\delta}$ is open. The rescaling $v^\lambda:=\lambda^av(\frac.{\lambda^b})$ implies that a neighborhood of zero is in $A_{a,b}^{+\delta}$. Moreover, this rescaling with $\lambda\rightarrow0$ gives that $A_{a,b}^{+\delta}$ is contracted to zero and so it is connected. Now, write $$A_{a,b}^{+\delta}=A_{a,b}^{+\delta}\cap( A_{a',b'}^{+\delta}\cup A_{a',b'}^{-\delta})=(A_{a,b}^{+\delta}\cap A_{a',b'}^{+\delta})\cup(A_{a,b}^{+\delta}\cap A_{a',b'}^{-\delta}).$$ 
Since by the definition, $A_{a,b}^{-\delta}$ is open and $0\in A_{a,b}^{+\delta}\cap A_{a',b'}^{+\delta}$, using a connectivity argument, we have $A_{a,b}^{+\delta}=A_{a',b'}^{+\delta}$. The proof is ended.
\end{proof}
{\bf Proof of Theorem \ref{t4}}. Using Lemma \ref{stb} via a translation argument, we can assume that $u(t)\in A_{a,b}^+$ for any $t\in[0,T^*)$. Moreover, thanks to the Lemma \ref{stbb}, we have $u(t)\in A_{1,1}^+$ for any $t\in[0,T^*)$. Taking account of the definition of $m$, we get
\begin{eqnarray*}
m
&>&S(u(t))\\
&>&S(u(t))-\frac1{2+N}K_{1,1}(u(t))\\
&=&\frac{\alpha }{2(2+N)}\|(-\Delta)^\frac\alpha2u(t)\|^2+\frac{\gamma+p-1}{(1+p)(2+N)}\int|x|^\gamma|u(t)|^{1+p}\,dx.
\end{eqnarray*}
This implies, via the conservation of the mass, that 
$$\sup_{[0,T^*]}\|u(t)\|_{H^\alpha}<\infty.$$
Then, $u$ is global.
\section{Appendix}
This section contains a proof of Lemma \ref{cmpct}. We suppose that $\gamma>0$, indeed $\gamma<0$ follows similarly and the case $\gamma=0$ is covered by Lemma \ref{sblv}.\\
Take $(u_n)$ a bounded sequence of $H^\alpha_{rd}$. Without loss of generality, we assume that 
$(u_n)$ converges weakly to zero in $H^\alpha$. Our purpose is to prove that $\|u_n\|_\Sigma\rightarrow0.$ Take $\varepsilon>0$ and write
\begin{eqnarray*}
\int|x|^\gamma|u_n|^{1+p}\,dx
&=&\Big(\int_{|x|\leq\varepsilon}+\int_{\varepsilon\leq|x|\leq\frac1\varepsilon}+\int_{|x|\geq\frac1\varepsilon}\Big)|x|^\gamma|u_n|^{1+p}\,dx.
\end{eqnarray*}
Since $\gamma>0$, taking account of Lemma \ref{gn}, we have
\begin{eqnarray*}
\int_{|x|\leq\varepsilon}|x|^\gamma|u_n|^{1+p}\,dx
&\leq&\varepsilon^\gamma\int|u_n|^{1+p}\,dx\\
&\leq&\varepsilon^\gamma\|u_n\|_{H^\alpha}^{1+p}\\
&\leq&C\varepsilon^\gamma.
\end{eqnarray*}
On the other hand, with Strauss inequality, via the fact that ${-\gamma+(p-1)({\frac N2-\alpha})}>0$ bcause $\alpha<\frac N2$, we get
\begin{eqnarray*}
\int_{|x|\geq\frac1\varepsilon}|x|^\gamma|u_n|^{1+p}\,dx
&=&\int_{|x|\geq\frac1\varepsilon}(|x|^{\frac N2-\alpha}|u_n|)^{-1+p}|x|^{\gamma-(p-1)({\frac N2-\alpha})}|u_n|^2\,dx\\
&\leq&C\|u_n\|_{\dot H^\alpha}^{p-1}\int_{|x|\geq\frac1\varepsilon}|x|^{\gamma-(p-1)({\frac N2-\alpha})}|u_n|^2\,dx\\
&\leq&C\|u_n\|_{H^\alpha}^{p+1}\varepsilon^{-\gamma+(p-1)({\frac N2-\alpha})}\\
&\leq&C\varepsilon^{-\gamma+(p-1)({\frac N2-\alpha})}.
\end{eqnarray*}
Now, by Rellich Theorem, it follows that when $n\longrightarrow\infty$,
$$\int_{\varepsilon\leq|x|\leq\frac1\varepsilon}|u_n|^{2}\,dx\longrightarrow0.$$
Moreover
\begin{eqnarray*}
\int_{\varepsilon\leq|x|\leq\frac1\varepsilon}|x|^\gamma|u_n|^{1+p}\,dx
&=&\int_{\varepsilon\leq|x|\leq\frac1\varepsilon}(|x|^{\frac N2-\alpha}|u_n|)^{-1+p}|x|^{\gamma-(p-1)({\frac N2-\alpha})}|u_n|^2\,dx\\
&\leq&C\|u_n\|_{\dot H^\alpha}^{p-1}\int_{\varepsilon\leq|x|\leq\frac1\varepsilon}|x|^{\gamma-(p-1)({\frac N2-\alpha})}|u_n|^2\,dx\\
&\leq&C\|u_n\|_{H^\alpha}^{p-1}\varepsilon^{-\gamma+(p-1)({\frac N2-\alpha})}\int_{\varepsilon\leq|x|\leq\frac1\varepsilon}|u_n|^2\,dx\\
&\leq&C\varepsilon^{-\gamma+(p-1)({\frac N2-\alpha})}.
\end{eqnarray*}
The proof is achieved when taking $\varepsilon$ tending to zero and $n$ going to infinity.

\end{document}